\documentclass[letterpaper, 10 pt, conference]{ieeeconf}  
\usepackage{amsmath} 
\usepackage{amssymb}  
\usepackage{graphicx}
\usepackage{epstopdf}
\usepackage{amsmath}
\usepackage{mathtools}
\usepackage{dsfont}
\usepackage{tikz}
\usepackage{siunitx}
\usepackage{xcolor}

\usepackage{hyperref}

\usepackage{balance}    
\usepackage{amsthm}
\def\sq{\mathbin{{\strut\rule{1.25ex}{1.25ex}}}}

\newcommand{\beq}{\begin{equation}}
\newcommand{\eeq}{\end{equation}}

\def\mathcolor#1#{\@mathcolor{#1}}
\def\@mathcolor#1#2#3{%
  \protect\leavevmode
  \begingroup
    \color#1{#2}#3%
  \endgroup
}

\newcounter{algorithmctr}[section]
\renewcommand{\thealgorithmctr}{\thesection.\arabic{algorithmctr}}
{\refstepcounter{algorithmctr}\begin{list}{}{%
\setlength{\rightmargin}{0\linewidth}%
\setlength{\leftmargin}{.05\linewidth}
\setlength{\itemsep}{1pt}
\setlength{\parskip}{0pt}
\setlength{\parsep}{0pt}}%
\rmfamily\small
\item[]{\setlength{\parskip}{0ex}\hrulefill\par%
\nopagebreak{\bfseries\textsf{Algorithm \thealgorithmctr~}}}}%
{{\setlength{\parskip}{-1ex}\nopagebreak\par\hrulefill} \end{list}}
\IEEEoverridecommandlockouts

\usepackage[noadjust]{cite} 
\usepackage{amsmath,stmaryrd,graphicx}

\newtheoremstyle{boldStyle}
  {\topsep}
  {\topsep}
  {\itshape}
  {0pt}
  {\bfseries}
  {.}
  { }
  {\thmname{#1}\thmnumber{ #2}\thmnote{ (#3)}}

\newtheoremstyle{italicStyle}
  {\topsep}
  {\topsep}
  {}
  {0pt}
  {\bfseries}
  {.}
  { }
  {\thmname{#1}\thmnumber{ #2}\thmnote{ (#3)}}

\theoremstyle{boldStyle}
\newtheorem{property}{Property}
\newtheorem{theorem}{Theorem}

\newtheorem{proposition}{Proposition}

\theoremstyle{italicStyle}
\newtheorem{assumption}{Assumption}
\newtheorem{remark}{Remark}

\renewenvironment{proof}{{\textbf{Proof:}}}{\hfill$\sq$}

\makeatletter
\newcommand{\fixed@sra}{$\vrule height 2\fontdimen22\textfont2 width 0pt\shortrightarrow$}
\newcommand{\shortarrow}[1]{%
  \mathrel{\text{\rotatebox[origin=c]{\numexpr#1*45}{\fixed@sra}}}
}
\makeatother

\title{\LARGE \bf
Multi-Rate Control Design  Leveraging Control Barrier Functions and Model Predictive Control Policies}

\author{Ugo Rosolia and Aaron D. Ames 
\thanks{Ugo Rosolia and Aaron D. Ames are with the AMBER lab at Caltech, Pasadena, USA.
E-mails: \tt\scriptsize{\{urosolia, ames\}@caltech.edu.}}
}%

\begin{document}

\maketitle
\thispagestyle{empty}
\pagestyle{empty}

\begin{abstract}
In this paper we present a multi-rate control architecture for safety critical systems. We consider a high level planner and a low level controller which operate at different frequencies. This multi-rate behavior is described by a piecewise nonlinear model which evolves on a continuous and a discrete level. First, we present sufficient conditions which guarantee recursive constraint satisfaction for the closed-loop system. Afterwards, we propose a control design methodology which leverages Control Barrier Functions (CBFs) for low level control and Model Predictive Control (MPC) policies for high level planning. The control barrier function is designed using the full nonlinear dynamical model and the MPC is based on a simplified planning model. When the nonlinear system is control affine and the high level planning model is linear, the control actions are computed by solving convex optimization problems at each level of the hierarchy. Finally, we show the effectiveness of the proposed strategy on a simulation example, where the low level control action is updated at a higher frequency than the high level command.
\end{abstract}

\section{Introduction}

Autonomous systems are designed to take control actions upon sensing the environment around them.
The decision making process is usually divided into different layers. For instance, in autonomous driving the top layer determines a goal or intention, such as lane keeping, merging or overtaking. Then, a high level planner computes a desired collision-free trajectory, which is then fed to a low level controller that computes the control action. Each layer operates at different frequency and it is designed using model of increasing accuracy and complexity. 

Combining high level planners with low level controllers has been extensively studied in literature~\cite{gurriet2018towards, CBF, wang2017safety, wabersich2018linear, herbert2017fastrack,yin2019optimization, singh2018robust, singh2017robust, gao2014tube, kogel2015discrete, yu2013tube}. Safety can be guaranteed using low level filters which, given a desired high level command, compute the closest safe control action using control barrier functions~\cite{gurriet2018towards,CBF,wang2017safety} or feasibility of an MPC problem~\cite{wabersich2018linear}.
The high level planner may be designed using a simplified model and the planned trajectory can be tracked using low level controllers.
The tracking error and the associated tracking policy can be computed using Hamilton-Jacobi (HJ) reachability analysis~\cite{herbert2017fastrack} or sum-of-squares programming~\cite{singh2018robust,yin2019optimization}.
Finally, high level planning and low level control can be implemented using nonlinear tube MPC strategies~\cite{gao2014tube, kogel2015discrete, yu2013tube, singh2017robust,kohler2020computationally}, where the difference between the planned trajectory and the actual one is over approximated using Lyapunov based analysis or Lipschitz properties of the nonlinear dynamics.

\begin{figure}[h!]
    \centering
	\includegraphics[trim= 10mm 8mm 2mm 1mm,width=0.95\columnwidth]{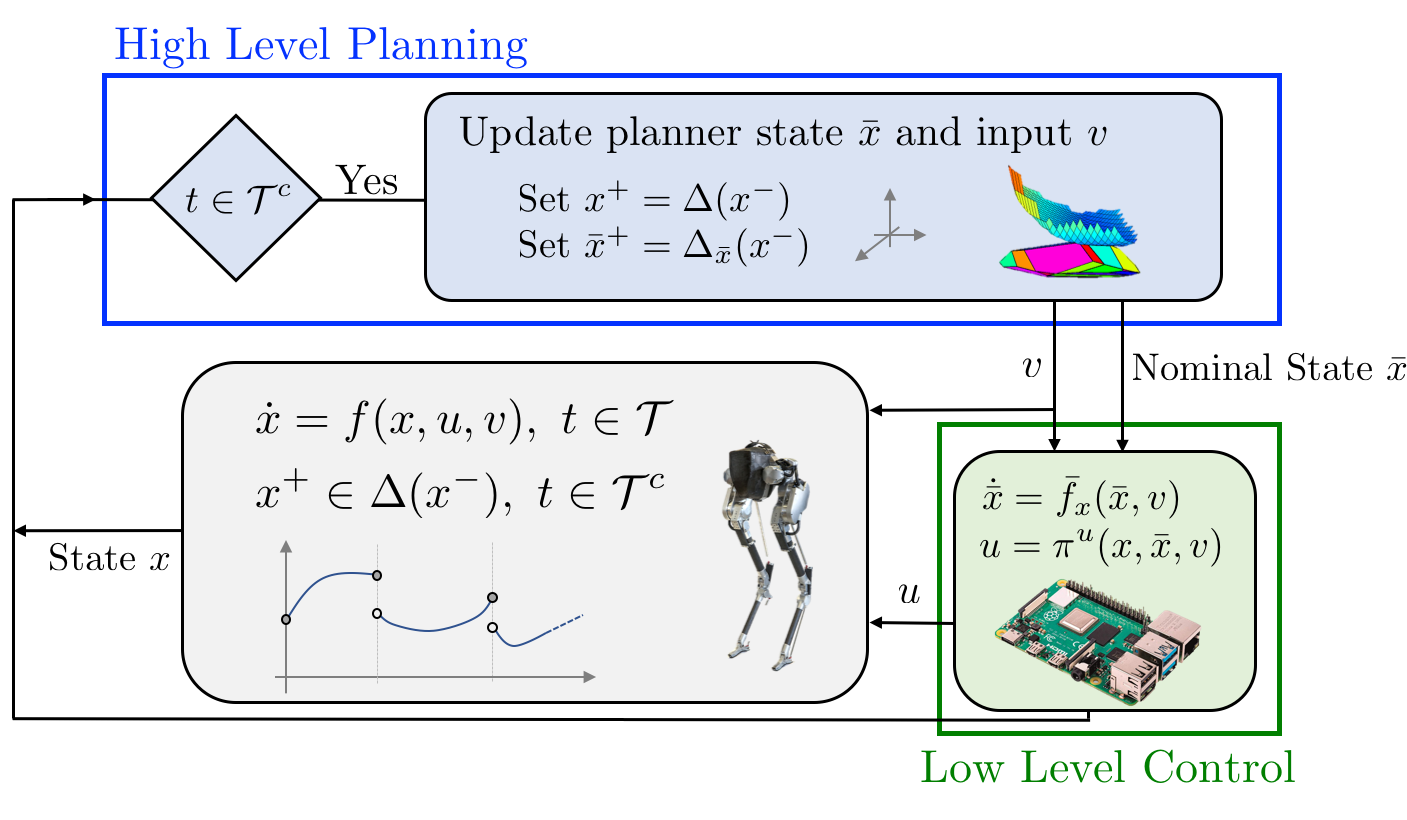}
    \caption{Representation of the multi-rate control architecture. The high level planner computes the desired state $\bar x$ and the high level command $v$.
    The low level controller computes at higher frequency the action $u$.}
    \label{fig:architecture}
\end{figure}

In the aforementioned papers, the low level and high level control actions are updated at the same frequency.
In this paper, we consider a high level planner which operates at a lower frequency than the low level controller. Multi-rate strategies are used in several applications, for instance in bipedal locomotion~\cite{luo2019robust, reher2019dynamic}, autonomous driving~\cite{rosolia2016autonomous, kapania2015design} and power grids~\cite{chen2019safety,farina2018hierarchical}. In this work, we introduce sufficient conditions to analyze the closed-loop safety properties of such control architectures.
Our contribution is threefold. 
First, we introduce sufficient conditions which guarantee recursive constraint satisfaction for a multi-frequency high level planning and low level control architecture, where the high level planner can reset its internal state as a function of the current state of the system.
Second, we present a control design which leverages CBFs for low level control and MPC for high level planning. We show that when the true system is nonlinear control affine and the planning model is linear, then the proposed strategy is implemented solving convex optimization problems. 
Third, we benchmark the proposed strategy against linear and nonlinear MPC policies. Simulation results demonstrate the benefit of the proposed multi-rate architecture, where the low level control action is updated at a higher frequency than the high level command. 

This paper is organized as follows. In Section~\ref{sec:problemFormulation} we introduce the problem formulation. Section~\ref{sec:frameworkArchitecture} describes the control architecture and the sufficient conditions which guarantee safety.
The synthesis process is described in Section~\ref{sec:controllerSyn} and it is demonstrated on a numerical example in Section~\ref{sec:Results}.\\
\textit{Notation:} The Minkowski sum of two sets $\mathcal{X}\subset \mathbb{R}^n$ and $\mathcal{Y}\subset \mathbb{R}^n$ is denoted as $\mathcal{X}\oplus\mathcal{Y}$, and the  Pontryagin difference as $\mathcal{X}\ominus\mathcal{Y}$. The set $\mathcal{K}^e$ is the set of extended class-$\mathcal{K}^e$ functions $\beta$ which are strictly increasing and $\beta(0)=0$.
Finally, given a function $f:\mathbb{R}^n \rightarrow \mathbb{R}^m$ and a set $\mathcal{X} \subset \mathbb{R}^n$ we denote the set $f(\mathcal{X}) = \{ y \in \mathbb{R}^m : \exists ~ x \in \mathcal{X} \text{ such that } y = f(x)\}$.

\section{Problem Formulation}\label{sec:problemFormulation}
This section introduces the system model and the synthesis objectives.
Consider a \textit{piecewise nonlinear model}:
\begin{equation}\label{eq:sys}
\begin{aligned}
    \Sigma : \begin{cases}\dot x(t) = f(x(t),u(t),v(t)), & t \in \mathcal{T}= \cup_{k=0}^\infty  (t_k, t_{k+1})\\
    x^+(t) = \Delta(x^-(t)), & t \in \mathcal{T}^c= \cup_{k=0}^\infty  \{t_k\} \end{cases},
\end{aligned}
\end{equation}
where the state $x\in \mathbb{R}^n$, the set $\mathcal{T}$ collects the open intervals from time $t_k$ to time $t_{k+1}$ and its complement $\mathcal{T}^c$ collects the time instances $t_k$. As a result, the above system~\eqref{eq:sys} evolves accordingly to the differential equation $\dot x(t) = f(x(t),u(t),v(t))$ between time $t_k$ and time $t_{k+1}$. On the other hand, at time $t_k$ the system evolution is defined by the \textit{reset map} $\Delta(\cdot)$, where $x^-(t) = \lim_{\tau \shortarrow{1} t}x(\tau)$ and $x^+(t) = \lim_{\tau \shortarrow{7} t}x(\tau)$ are the right and left limits of a trajectory $x(t)$ which is assumed right continuous. 
Furthermore, we assume that the input $u \in \mathbb{R}^d$ is a continuous function of the state $x$ and the input $v \in \mathbb{R}^d$ is a piecewise-constant function which is updated when $t \in \mathcal{T}^c$, i.e.,
\begin{equation}\label{eq:controlPolicy}
\begin{aligned}
    \Pi : \begin{cases}u(t) = \pi^u(x(t),v(t)),~ \dot v(t) = 0, & t \in \mathcal{T}\\
    u^+(t) = u^-(t),~  v^+(t) = \pi^v(x^+(t)), & t \in \mathcal{T}^c\end{cases}.
\end{aligned}    
\end{equation}
The above control policies~\eqref{eq:controlPolicy} in closed-loop with system~\eqref{eq:sys} results in a piecewise nonlinear autonomous system, which evolves on a discrete and a continuous level.

\textbf{Objective: }Our goal is to steer the system from a starting state $x_s$ to a goal state  $x_g$ while satisfying the following state and input constraints:
\begin{equation}\label{eq:stateInputConstr}
    \begin{aligned}
        & x(t) \in \mathcal{X}_c \subset \mathbb{R}^n,~\forall  t \in \mathcal{T},\\
        & x^+(t) \in \mathcal{X}_d \subseteq \mathcal{X}_c\subset \mathbb{R}^n, ~\forall t \in \mathcal{T}^c, \\
        & u(t) \in \mathcal{U}\subset \mathbb{R}^d, v(t) \in \mathcal{V} \subset \mathbb{R}^d, \forall t \geq 0.
    \end{aligned}
\end{equation}

\section{Framework Architecture And Properties }\label{sec:frameworkArchitecture}
In this section we present the multi-rate control architecture. 
First, we introduce an augmented model, which is composed by the piecewise nonlinear system~\eqref{eq:sys} and a high level planning model. 
The latter is affected by the piecewise constant input $v$ and it is used to compute the planner state $\bar x \in \mathbb{R}^n$.
Afterwards, the planned trajectory together with the input $v$ are fed to the low level controller which computes the control action $u$, as shown in Figure~\ref{fig:architecture}.

\subsection{Augmented System}
The augmented system is defined as
\begin{equation}\label{eq:augSys}
\begin{aligned}
    \Sigma_{\bar x} : \begin{cases}
    \begin{matrix*}[l]\dot x(t) = f\big(x(t),u(t),v(t)\big) \\\dot{\bar x}(t) = f_{\bar x}(\bar x(t),v(t)) \end{matrix*}, &  t \in \mathcal{T}= \cup_{k=0}^\infty  (t_k, t_{k+1})\\
    \begin{matrix*}[l]x^+(t) = \Delta(x^-(t)) \\ \bar x^+(t) = \Delta_{\bar x}(x^-(t)) \end{matrix*}, &  t \in \mathcal{T}^c= \cup_{k=0}^\infty  \{t_k\}\\
    \end{cases}
\end{aligned}
\end{equation}
and the control actions are given by the policies
\begin{equation}\label{eq:augPolicy}
\begin{aligned}
    \Pi_{\bar x} : \begin{cases} u(t) = \pi^u\big(x(t), \bar x(t), v(t)\big),~ \dot v(t) = 0, & t \in \mathcal{T}\\
     u^+(t) = u^-(t),~  v^+(t) = \pi^v\big(x^+(t)\big), & t \in \mathcal{T}^c\end{cases},
\end{aligned}    
\end{equation}
where $\mathcal{T}$ and $\mathcal{T}^c$ are defined as in~\eqref{eq:augSys} and $\bar x \in \mathbb{R}^n$ represents the planned state which is affected by the piecewise input $v$. 


\subsection{High Level and Low Level Properties}\label{sec:properties}
In this section, we define four properties associated with the high level planner and low level controller. 
As we will discuss later on, when these properties hold the closed-loop system is guaranteed to recursively satisfy state and input constraints~\eqref{eq:stateInputConstr}. 

Consider the closed-loop system~\eqref{eq:augSys}-\eqref{eq:augPolicy} and let $t_k$ be the time at which the $k$th discontinuous transition occurs, i.e., $  x^+(t_k)= \Delta (x^-(t_k))~ \forall k \in \{1,2,\ldots \}$. 
We define the error $e=x-\bar x$ and we introduce the following error dynamics:
\begin{equation}\label{eq:errorDef}
\begin{aligned}
    \Sigma_e : \begin{cases}\dot e(t) = f_e\big(x(t),u(t),v(t),\bar x(t)\big), & t \in \mathcal{T} \\
    e^+(t) = \Delta_e(x^-(t)), & t \in \mathcal{T}^c \end{cases},
\end{aligned}    
\end{equation}
where $\mathcal{T}$ and $\mathcal{T}^c$  are defined as in~\eqref{eq:augSys}, and the reset map $\Delta_e(\cdot)$ is designed so that the following properties hold.

\begin{property}[\textbf{low level safety}]\label{prop:lls}
The control policy $\pi^u(\cdot)$ from~\eqref{eq:augPolicy} guarantees \textit{low level safety} for the closed-loop system~\eqref{eq:augSys}-\eqref{eq:augPolicy} and the set $\mathcal{S}_x \subseteq \mathcal{X}_c \subset \mathbb{R}^n$, if $ \forall x^+(t_k) \in \mathcal{S}_x \cap \mathcal{X}_d$ and $\forall v^+(t_k) \in \mathcal{V}$ we have that
\begin{equation}
    \begin{aligned}
    x(t) \in \mathcal{S}_x \subseteq \mathcal{X}_c \textit{ and }u(t) \in \mathcal{U}, \forall t \in (t_k, t_{k+1}).
    \end{aligned}
\end{equation}
\end{property}

Basically, the above property guarantees that state and input constraints are satisfied when the system evolves smoothly between time $t_{k}$ and time $t_{k+1}$. In particular, if at time $t_k$ the state $x^+(t_k)$ belongs to the set $\mathcal{S}_x \cap \mathcal{X}_d $, then the low level controller $\pi^u(\cdot)$ guarantees state and input constraint satisfaction until the next discontinuous transition at time $t_{k+1}$.

\begin{property}[\textbf{low level tracking}]\label{prop:llt}
The control policy $\pi^u(\cdot)$ from~\eqref{eq:augPolicy} guarantees \textit{low level tracking} for the closed-loop system~\eqref{eq:augSys}-\eqref{eq:augPolicy}, the set $\mathcal{S}_e \subset \mathbb{R}^n$ and  the set $\mathcal{S}_x \subseteq \mathcal{X}_c \subset \mathbb{R}^n$, if $ \forall e^+(t_k) = x^+(t_k) - \bar x^+(t_k) \in \mathcal{S}_e$, $\forall x^+(t_k) \in \mathcal{S}_x \cap \mathcal{X}_d$ and $\forall v^+(t_k) \in \mathcal{V}$ we have that
\begin{equation}
    \begin{aligned}
        e(t) = x(t) - \bar x(t) \in \mathcal{S}_e, \forall t \in (t_k, t_{k+1}).
    \end{aligned}
\end{equation}
\end{property}

The low level tracking property ensures that the difference between the planned trajectory and the true state is contained into the set $\mathcal{S}_e$ for all time $t \in (t_k, t_{k+1})$. The above Properties~\ref{prop:lls}-\ref{prop:llt} guarantee that the planned trajectory can be safely executed by the true system.

\begin{property}[\textbf{high level safety}]\label{prop:hls}
The control policy $\pi^v(\cdot)$ from~\eqref{eq:augPolicy} guarantees high level safety for the closed-loop system~\eqref{eq:augSys}-\eqref{eq:augPolicy}, the set $\mathcal{S}_e \subset \mathbb{R}^n$ and  the set $\mathcal{S}_x \subseteq \mathcal{X}_c \subset \mathbb{R}^n$, if for the initial conditions $x(0) = \bar x(0) + e(0) \in \mathcal{S}_x \cap \mathcal{X}_d$ and $e(0) \in \mathcal{S}_e$ we have that $\pi^v(x^+(0)) \in \mathcal{V}$ and
\begin{equation}\label{eq:hls}
\begin{aligned}
    & z \in \mathcal{S}_x \cap \mathcal{X}_d,\\
    &\pi^v(z) \in \mathcal{V},~\forall z \in \Delta (\{\bar x^-(t_k)\}\oplus\mathcal{S}_e),\forall k \in \{1,2, \ldots\}.
   \end{aligned}
\end{equation}
\end{property}

\begin{property}[\textbf{high level tracking}]\label{prop:hlt}
The reset map $ \Delta_e(\cdot)$ from~\eqref{eq:errorDef} guarantees high level tracking for the closed-loop system~\eqref{eq:augSys}-\eqref{eq:augPolicy}, the set $\mathcal{S}_e \subset \mathbb{R}^n$ and  the set $\mathcal{S}_x \subseteq \mathcal{X}_c \subset \mathbb{R}^n$, if for the initial conditions $x(0) = \bar x(0) + e(0) \in \mathcal{S}_x \cap \mathcal{X}_d$ and $e(0) \in \mathcal{S}_e$ we have that
\begin{equation}
\begin{aligned}
    &\Delta(z) = \Delta_{\bar x}(z) + \Delta_e(z),\\
    &\Delta_e(z) \in \mathcal{S}_e, \forall z \in \{ \bar x^-(t_k) \} \oplus\mathcal{S}_e, \forall k \in \{0,1,\ldots\}.
    \end{aligned}
\end{equation}
\end{property}

It is important to underline that the above Properties~\ref{prop:hls}-\ref{prop:hlt} are defined for the planning model at the discontinuous transitions, i.e., at time $t_k$ for $k \in \{0,1,\ldots\}$. This fact allows us to design a high level planner based on a discrete time model, which describes the evolution of the planned trajectory from $\bar x^+(t_k)$ to $\bar x^+(t_{k+1})$.

\subsection{Safety Guarantees}\label{sec:safetyGuarantees}

In this section, we show that when the control policies from~\eqref{eq:augPolicy} satisfy Properties~\ref{prop:lls}-\ref{prop:hlt}, the closed-loop system~\eqref{eq:augSys}-\eqref{eq:augPolicy} does not violate state and input constraints~\eqref{eq:stateInputConstr}. 

\begin{theorem}\label{th:safety}
 Assume that Properties~\ref{prop:lls}-\ref{prop:hlt} are satisfied for the closed-loop system~\eqref{eq:augSys}-\eqref{eq:augPolicy}, the set $\mathcal{S}_e \subset \mathbb{R}^n$ and the set $\mathcal{S}_x \subseteq \mathcal{X}_c \subset \mathbb{R}^n$. Let $x(0) = \bar x(0) + e(0) \in \mathcal{S}_x \cap \mathcal{X}_d$ and $e(0) \in \mathcal{S}_e$. 
Then, the closed-loop system~\eqref{eq:augSys}-\eqref{eq:augPolicy} satisfies state and input constraints~\eqref{eq:stateInputConstr} for all time $t\geq0$.
\end{theorem}

\begin{proof}
The proof proceeds by induction. 
Assume that after the $k$th discontinuous transition $x^+(t_k) \in \mathcal{S}_x \cap \mathcal{X}_d$, $e^+(t_k) = x^+(t_k) - \bar x^+(t_k) \in \mathcal{S}_e$ and $v^+(t_k)\in \mathcal{V}$, then by Property~\ref{prop:lls} 
\begin{equation}\label{eq:proof:llc}
    x(t) \in \mathcal{S}_x \subseteq \mathcal{X}_c,u(t)\in \mathcal{U}, ~\forall t \in (t_k, t_{k+1}).
\end{equation}
Furthermore, by Property~\ref{prop:llt} we have that at time $t_{k+1}$
\begin{equation*}
    x^-(t_{k+1}) \in \{{\bar x}^-(t_{k+1}) \} \oplus \mathcal{S}_e.
\end{equation*}
The above equation together with Property~\ref{prop:hls} implies that
\begin{equation}
\begin{aligned}
    &x^+(t_{k+1}) = \Delta(x^-(t_{k+1})) \in \mathcal{S}_x \cap \mathcal{X}_d, v(t_{k+1}) \in \mathcal{V}.
\end{aligned}
\end{equation}
Finally, from Property~\ref{prop:hlt} we have 
\begin{equation}\label{eq:proof:hlt}
    \begin{aligned}
        &e^+(t_{k+1}) =  x^+(t_{k+1}) -\bar x^+(t_{k+1}) \in \mathcal{S}_e.
    \end{aligned}
\end{equation}
The above equations~\eqref{eq:proof:llc}-\eqref{eq:proof:hlt} imply that, if $x^+(t_k) \in \mathcal{S}_x \cap \mathcal{X}_d$, $e^+(t_k) = x^+(t_k) - \bar x^+(t_k) \in \mathcal{S}_e$ and $v^+(t_k)\in \mathcal{V}$, then state and input constraints~\eqref{eq:stateInputConstr} are satisfied for all $t\in(t_k, t_{k+1})$. Furthermore, we have that the state $x^+(t_{k+1}) \in \mathcal{S}_x  \cap \mathcal{X}_d$, the error $e^+(t_{k+1}) = x^+(t_{k+1}) - \bar x^+(t_{k+1}) \in \mathcal{S}_e$ and the input $v^+(t_{k+1})\in\mathcal{V}$. \\
Finally, by assumption $x(0) = \bar x(0) + e(0) \in \mathcal{S}_x \cap \mathcal{X}_d$ and $e(0) \in \mathcal{S}_e$, which imply from Property~\ref{prop:hls} that $v(0)\in\mathcal{V}$. 
Therefore, from equations~\eqref{eq:proof:llc}-\eqref{eq:proof:hlt}, we conclude by induction that the closed-loop system~\eqref{eq:augSys}-\eqref{eq:augPolicy} recursively satisfies state and input constraints for all $t\geq 0$.
\end{proof}

\begin{remark}
    We underline that guarantees from Theorem~\ref{th:safety} hold when the control action $u(t)$ is updated continuously. However, in practice the control action is updated at a high frequency, for instance at $1kHz$ in our simulations.
\end{remark}

\section{Synthesis: Leveraging CBFs and MPC}\label{sec:controllerSyn}

In this section, we discuss how the properties from Section~\ref{sec:properties} may be used to synthesize a safe controller. First, we show that Control Barrier Functions (CBFs) may be used to enforce low level safety and low level tracking. Afterwards, we design a Model Predictive Controller (MPC) to enforce high level safety and high level tracking.

We consider a control affine system where the input is given by the summation of the continuous control action $u$ and the piecewise constant action $v$, i.e.,
\begin{equation}\label{eq:controlAffineSystem}
\begin{aligned}
    \Sigma^a : \begin{cases}\dot x(t) = f\big(x(t)\big) + g\big(x(t)\big)\big(u(t)+v(t)\big), & t \in \mathcal{T}\\
    x^+(t) = \Delta\big(x^-(t)\big), & t \in \mathcal{T}^c\end{cases}
\end{aligned}
\end{equation}
where $\mathcal{T} = \cup_{k=0}^\infty(kT,(k+1)T)$, $\mathcal{T}^c = \cup_{k=0}^\infty \{kT\}$ and $1/T$ is the frequency at which the high level command is updated. Furthermore, we assume the $f$ and $g$ are locally Lipschitz continuous with respect to their arguments and that the map $\Delta(\cdot)$ is affine, as stated in Assumption~\ref{ass:affineMap}. 
Finally, the augmented model is given by
\begin{equation}\label{eq:augControlAffine}
\begin{aligned}
    \Sigma^a_{\boldsymbol{z}} : \begin{cases}\boldsymbol{\dot z}=\begin{bmatrix*}[l] \dot x \\ \dot {\bar x} \end{bmatrix*} = f^a(\boldsymbol{z}) + g^a(\boldsymbol{z})(u+v), 
    & \!\! t \in \mathcal{T}\\
     \boldsymbol{z}^+ = \begin{bmatrix}x^+ \\ \bar x^+ \end{bmatrix} = \begin{bmatrix} \Delta(x^{-}) \\ \Delta_{\bar x}(\bar x^{-}) \end{bmatrix}, & \!\! t \in \mathcal{T}^c\end{cases}.
\end{aligned}
\end{equation}
where we dropped the dependence on time $t$ to simplify the notation and the continuous evolution of the planning state $\bar x$ is described by linear dynamics, i.e., $\dot{\bar{x}}=A\bar x +Bv$.

\begin{assumption}\label{ass:affineMap} 
The functions $f$ and $g$ are locally Lipschitz continuous and the reset map $\Delta(\cdot)$ from \eqref{eq:controlAffineSystem} is affine. Consequently, the reset map $\Delta$ can be written as $\Delta(x) = Tx + p$, for some matrix $T\in \mathbb{R}^{n\times n}$ and some vector $p \in \mathbb{R}^{n}$.
\end{assumption}

\subsection{Control Barrier Functions}
In this section, we show that CBFs~\cite{CBF} can be used to enforce low level safety and low level tracking. Furthermore, we introduce a Control Lyapunov Function (CLF) which is used to reduce the tracking error. Finally, we combine CFBs and CLF into a QP, which defines the low level control policy from Figure~\ref{fig:architecture}.

First we define the following sets:
\begin{equation}\label{eq:S_x_S_e}
\begin{aligned}
    & \mathcal{S}_x = \big\{ x \in \mathbb{R}^{n} : h_x(x) \geq 0 \big\} \subseteq \mathcal{X}_c \subset \mathbb{R}^n, \\
    & \mathcal{S}_e = \{ e \in \mathbb{R}^{n} : h_e(e) \geq 0 \} \subset \mathbb{R}^n,
\end{aligned}
\end{equation}
which will be used to check if Properties~\ref{prop:lls}-\ref{prop:hlt} hold. The above functions $h_x$ and $h_e$ are designed by the user based on the application, as shown in the result section.
Furthermore, we define $||x||_Q = x^\top Qx$ and we introduce the candidate Lyapunov function
\begin{equation}\label{eq:CLF}
    V(\boldsymbol{z} )  = ||x - \bar x||_{Q_v},
\end{equation}
which penalizes the deviation of the true system from the planned trajectory.

Finally, the CBFs associated with the sets in~\eqref{eq:S_x_S_e}, and the CLF~\eqref{eq:CLF} are used to define the following CLF-CBF Quadratic Program (QP):
\begin{equation}\label{eq:CLF-CBF-QP}
    \begin{aligned}
        \min_{u \in \mathcal{U}, \gamma} ~ &||u||_2 + c_1 \gamma^2 \\
        \text{s.t. } ~ & \frac{\partial V(\boldsymbol{z})}{\partial \boldsymbol{z}}(f^a(\boldsymbol{z}) + g^a(\boldsymbol{z})(v+u)) \leq - c_2V+ \gamma \\
        & \frac{\partial h_{x}(x)}{\partial x}(f(x) + g(x)(v+u))  \geq - \alpha_1(h_{x}) \\
        & \frac{\partial h_{e}(e )}{\partial \boldsymbol{z}}(f^a(\boldsymbol{z}) + g^a(\boldsymbol{z})(v+u))  \geq - \alpha_2(h_{e}).
    \end{aligned}
\end{equation}
where $e = x - \bar x \in \mathbb{R}^n$ and $\boldsymbol{z} = [x^\top, \bar x^\top]^\top \in \mathbb{R}^{2n}$. Furthermore, in the above QP $c_1>0$, $c_2>0$, $\alpha_1 \in \mathcal{K}^e$ and $\alpha_2 \in \mathcal{K}^e$. Let $u^*(x,\bar x, v)$ and $\gamma^*(x,\bar x, v)$ be the optimal solution to~\eqref{eq:CLF-CBF-QP}, the low level policy is defined as
\begin{equation}\label{eq:uPolicy}
    \pi^{u}(x, \bar x, v ) = u^*(x,\bar x, v).
\end{equation}

\begin{assumption}\label{ass:QPfeasibility}
The Quadratic Program (QP)~\eqref{eq:CLF-CBF-QP} is feasible for all $\boldsymbol{z} \in \mathcal{I}=\{\boldsymbol{z}=[x^\top,\bar x^\top]^\top \in \mathbb{R}^{2n}: h_e(x-\bar x) \geq 0, h_x(x)\geq 0 \}$ and for all $v \in \mathcal{V}$.
\end{assumption}

\begin{proposition}\label{proposi:llProp}
Consider the system~\eqref{eq:augControlAffine} and \eqref{eq:uPolicy} with $v(t)\in\mathcal{V},\forall t\geq0$. If Assumptions~\ref{ass:affineMap}-\ref{ass:QPfeasibility} hold, then the control policy \eqref{eq:uPolicy} guarantees that Properties~\ref{prop:lls} and~\ref{prop:llt} are satisfied for the sets $\mathcal{S}_x$ and $\mathcal{S}_e$ from~\eqref{eq:S_x_S_e}.
\end{proposition}
\begin{proof}
The proof follows from~\cite{CBF}.
\end{proof}

\begin{remark}
    We underline that Assumption~\ref{ass:QPfeasibility} is satisfied for some $\alpha_1 \in \mathcal{K}^e$ and $\alpha_2 \in \mathcal{K}^e$ when the set $\mathcal{I}$ is robust control invariant for system~\eqref{eq:augSys} and mild assumptions on the Lie derivative of~\eqref{eq:augSys} hold (see \cite{CBF} for further details). The set $\mathcal{I}$ may be hard to compute and standard techniques are based on HJB reachability analysis~\cite{herbert2017fastrack}, SOS programming~\cite{singh2018robust}, Lyapunov-based methods~\cite{singh2017robust} and Lipschitz properties of the system dynamics~\cite{chen2018data,yu2013tube}.
\end{remark}

\subsection{Discrete Uncertain Model}
The CLF-CBF QP~\eqref{eq:CLF-CBF-QP} computes a control action which constraints the difference between the planned trajectory and the true system into $\mathcal{S}_e$. In this section, we leverage this property to construct a discrete time linear uncertain model, which over-approximates the evolution of the true system from $x^+(t_k)$ to $x^+(t_{k+1})$.

First, we define the following reset maps for the error and planning dynamics from~\eqref{eq:augControlAffine}:
\begin{equation}\label{eq:errorDynForCLF-CBF}
\begin{aligned}
    & \bar x^+ = \Delta_{\bar x}(x^-) = \Delta(x^-) = x^+\\
    & e^+ = \Delta_e(e^-) = 0.
\end{aligned}
\end{equation}
Basically, the above reset maps set the planning state $\bar x$ equal to the true state $x$, and consequently the error state $e=0$ after each $k$th discontinuous transition. We underline that setting $\bar x^+ = x^+$ is a design choice. It would be possible to design $\Delta_{\bar x}$, $\Delta$ and $\Delta_e$ such that $x^+ - \bar x^+ = e^+$ and let the high-level planner to select $\bar x^+$.

As the planning model is linear for all $t \in (t_k, t_{k+1})$, we have that
\begin{equation}\label{eq:nominalDiscreteModel}
\begin{aligned}
    & \bar x^-(t_{k+1}) = \bar A \bar x^+(t_k) + \bar B v^+(t_k),\\
\end{aligned}
\end{equation}
where the transition matrices are $\bar A = e^{A T} \text{ and } \bar B = \int_0^{T} e^{A(T-\eta)}B d\eta$. 
We notice that, when Assumptions~\ref{ass:affineMap}-\ref{ass:QPfeasibility} hold, from Proposition~1 we have that $x^-(t_{k+1}) \in \{ \bar x^-(t_{k+1})\} \oplus \mathcal{S}_e$. Furthermore, from equations~\eqref{eq:errorDynForCLF-CBF}-\eqref{eq:nominalDiscreteModel} we have that
\begin{equation}\label{eq:nominalDiscreteModelSteps}
\begin{aligned}
    x^+(t_{k+1}) &= \Delta(x^-(t_{k+1})) \in \Delta(\{ \bar x^-(t_{k+1}) \} \oplus \mathcal{\bar S}_e) \\
    &=  \{ T\bar A \bar x^+(t_k) + T \bar B v^+(t_k) + p \} \oplus T\mathcal{\bar S}_e \\
    &= \{ T \bar A x^+(t_k)) + T \bar B v^+(t_k) \} \oplus \Delta(\mathcal{\bar S}_e),
\end{aligned}
\end{equation}
where the polytope  $\mathcal{\bar S}_e$ contains the set $\mathcal{S}_e$, i.e.,  ${\mathcal{S}}_e \subseteq \bar{\mathcal{S}}_e$. Equation~\eqref{eq:nominalDiscreteModelSteps} defines a discrete time uncertain linear system which can be used to check if Property~\ref{prop:hls} is satisfied, as stated by the following proposition. Notice that in~\eqref{eq:nominalDiscreteModelSteps} we used the definition of $\Delta(\cdot)$ from Assumption~\ref{ass:affineMap}.


\begin{proposition}\label{proposi:hlProp}
Let Assumptions~\ref{ass:affineMap}-\ref{ass:QPfeasibility} hold. Consider the autonomous discrete time uncertain system
\begin{equation}\label{eq:linearDiscreteSystem}
    \begin{aligned}
        & x_d^+({k+1}) = T \bar A x_d^+(k) + T \bar B \pi^v(x_d^+(k)) + \bar w_d^+(k)\\
        & x_d^+(0)= x(0)= \bar x(0) = x_0,
    \end{aligned}
\end{equation}
where the control policy $\pi^v:\mathbb{R}^n \rightarrow \mathcal{V}$ and the disturbance $\bar w_d^+(k) \in \Delta(\bar{\mathcal{S}}_e)$. If the state of the above system $x_d^+(k) \in \mathcal{X}_d \cap \mathcal{S}_x$, $\forall \bar w_d^+(k) \in \Delta(\mathcal{\bar S}_e)$ and $\forall k \in \{0,1,\ldots\}$. Then Property~\ref{prop:hls} is satisfied  for the sets $\mathcal{S}_x$ and $\mathcal{S}_e$ from~\eqref{eq:S_x_S_e} and system~\eqref{eq:augControlAffine} in closed-loop with
\begin{equation}\label{eq:prop2Policy}
\begin{aligned}
    \Pi_{\bar x}^v : \begin{cases} u(t) = \pi^u\big(x(t), \bar x(t), v(t)\big),~ \dot v(t) = 0, & t \in \mathcal{T}\\
     u^+(t) = u^-(t),~  v^+(t) = \pi^v\big(x^+(t)\big), & t \in \mathcal{T}^v\end{cases},
\end{aligned}    
\end{equation}
where $\pi^u$ is defined in~\eqref{eq:uPolicy} and $\pi^v$ is the control policy from~\eqref{eq:linearDiscreteSystem}.
\end{proposition}

\begin{proof}
First, we recursively define the $k$-steps robust reachable sets for the discrete time autonomous uncertain system~\eqref{eq:linearDiscreteSystem} and for $k\in \{0,1,\ldots\}$ 
\begin{equation*}
\begin{aligned}
    \mathcal{R}_{k+1} =\{x^+ \in \mathbb{R}^n:& \exists x \in \mathcal{R}_{k}, \exists w \in \Delta(\mathcal{\bar{S}}_e),\\
    &x^+ = T \bar A x + T \bar B \pi^v(x) + w \},
\end{aligned}
\end{equation*}
where $\mathcal{R}_0 = \{x_0\}$. Notice that by assumption $\mathcal{R}_{k}\subset \mathcal{X}_d\cap\mathcal{S}_x,~\forall k\in \{0,1,\ldots\}$, $\pi^v:\mathbb{R}^n\rightarrow \mathcal{V}$, $x(0)=\bar x(0) = x_0$ and $e(0)=0$. Finally, Assumptions~\ref{ass:affineMap}-\ref{ass:QPfeasibility} and equations~\eqref{eq:nominalDiscreteModelSteps}-\eqref{eq:linearDiscreteSystem} imply that the state $x(t)$ of the closed-loop system~\eqref{eq:augControlAffine} and~\eqref{eq:prop2Policy} satisfies $x^+(t_k) \in \mathcal{R}_k \subset \mathcal{X}_d \cap \mathcal{S}_x,~\forall k\in\{0, 1,\ldots\}.$
Therefore,  Property~\ref{prop:hls} is satisfied  for the sets $\mathcal{S}_x$ and $\mathcal{S}_e$ from~\eqref{eq:S_x_S_e} and the closed-loop system~\eqref{eq:augControlAffine} and~\eqref{eq:prop2Policy}.  
\end{proof}

\subsection{Model Predictive Control}\label{sec:MPC}
In this section, we design a Model Predictive Controller that allows us to guarantee high level safety and high level tracking from Properties~\ref{prop:hls}-\ref{prop:hlt}. In particular, we leverage the result from Proposition~\ref{proposi:hlProp} and we design a robust tube MPC with time-varying cross section as in~\cite{chisci2001systems}.


At time $t_k$ given the state of the system $x(t_k)$ we solve the following finite time  optimal control problem:
\begin{equation}\label{eq:ftocp}
\begin{aligned}
    \min_{\boldsymbol{v}_t} \quad & \sum_{k = t}^{t+N} \Big( || x_{k|t}-x_g||_Q + ||v_{k|t}||_R \Big) +|| x_{t+N|t}-x_g||_{Q_f} \\
    \text{s.t.} \quad & x_{k+1|t} = T (\bar A + \bar B K) x_{k|t} + T \bar B v_{k|t}\\
    &  x_{t|t} = x^+(t_k)\\
    &  x_{k|t} \in \mathcal{X}_d \cap \mathcal{S}_x \ominus \mathcal{E}_k, ~ v_{k|t} \in \mathcal{V} \ominus K \mathcal{E}_k  \\
    &  x_{t+N|t} \in \mathcal{X}_F \ominus \mathcal{E}_{t+N},\forall k = \{t, \ldots, t+N \}
\end{aligned}
\end{equation}
where $K$ is a stabilizing feedback gain, $||x||_Q = x^\top Qx$,  $\mathcal{E}_{k+1} = T (\bar A + \bar BK) \mathcal{E}_{k} \oplus \Delta(\mathcal{\bar S}_e)\text{ and }\mathcal{E}_0 = \{0\}.$
The above control problem computes a sequence of open loop actions $\boldsymbol{v}_t=[v_{t|t},\ldots,v_{t+N|t}]$ which robustly steer system~\eqref{eq:linearDiscreteSystem} from the current state $x(t_k)$ to the terminal set $\mathcal{X}_F$, while minimizing the nominal cost and robustly satisfying state and input constraints~\cite{chisci2001systems}.
Let $\boldsymbol{v}_t^*=[v^{*}_{t|t},\ldots,v^{*}_{t+N|t}]$ be the optimal solution and $[x^{*}_{t|t},\ldots,x^{*}_{t+N|t}]$ the associated optimal trajectory, then the MPC policy is 
\begin{equation}\label{eq:vPolicy}
    \pi^{v}\big(x^+(t_k)\big) = v^{*}_{t|t} + Kx_{t|t}^* .
\end{equation}

\begin{assumption}\label{ass:mpcAssumption}
The terminal constraint set $\mathcal{X}_F \subset \mathcal{X}_d$ in~\eqref{eq:ftocp} is a robust positive invariant set for the discrete time uncertain autonomous system $x(k+1) = T(\bar A + \bar B K)x(k) + w(k)$ with $w(k) \in \Delta(\mathcal{\bar S}_e)$ for all $k\in \{0,1,\dots\}$.
\end{assumption}

\subsection{Closed-loop Properties}\label{sec:closedProperties}
In this section, we show that Properties~\ref{prop:lls}-\ref{prop:hlt} hold for the closed-loop system~\eqref{eq:augControlAffine},~\eqref{eq:uPolicy} and~\eqref{eq:vPolicy} and the sets in~\eqref{eq:S_x_S_e}. Therefore, the closed-loop system satisfies state and input constraints~\eqref{eq:stateInputConstr}.

\begin{theorem}
Consider system~\eqref{eq:augControlAffine} in closed-loop with~\eqref{eq:uPolicy}  and~\eqref{eq:vPolicy}. Let Assumptions~\ref{ass:affineMap}-\ref{ass:mpcAssumption} hold. Assume that problem~\eqref{eq:ftocp} is feasible at time $t=0$, then the closed-loop system~\eqref{eq:augControlAffine},~\eqref{eq:uPolicy}  and~\eqref{eq:vPolicy} satisfies state and input constraints~\eqref{eq:stateInputConstr} for all time $t \geq 0$. 
\end{theorem}
\begin{proof}
Notice that from Proposition~\ref{proposi:llProp} and equation~\eqref{eq:errorDynForCLF-CBF}, we have that the closed-loop system satisfies Properties~\ref{prop:lls},~\ref{prop:llt} and~\ref{prop:hlt}. Moreover from standard MPC arguments~\cite{chisci2001systems, borrelli2017predictive}, we have that
the closed-loop system~\eqref{eq:linearDiscreteSystem}, where the control policy $\pi^v(\cdot)$ is the MPC policy~\eqref{eq:vPolicy}, evolves inside~$\mathcal{X}_d \cap \mathcal{S}_x$ and $\pi^v(x(t_k)) \in \mathcal{V}$ for all $k \in \{0,1,\ldots\}$ (thus Property~\ref{prop:hls} hold). Concluding, Properties~\ref{prop:lls}-\ref{prop:hlt} are guaranteed for the closed-loop system~\eqref{eq:augControlAffine},~\eqref{eq:uPolicy}  and~\eqref{eq:vPolicy} and state and input constraints~\eqref{eq:stateInputConstr} are satisfied for all time $t \geq 0$.
\end{proof}

\section{Simulation Results}\label{sec:Results}
We use the proposed strategy to steer a Segway to a goal state\footnote{\scriptsize{Code available at \url{https://github.com/urosolia/MultiRate}}}, as shown in  Figure~\ref{fig:segway}.
The state of the system are the position $p_x$, the velocity $v_x$, the rod angle $\theta$ and the angular velocity $\omega$. The control action is the voltage commanded to the motor and the equations of motion used to simulate the system can be found in~\cite[Section~IV.B]{gurriet2018towards}. The nominal model is obtained using a small angle approximation and the MPC is implemented for $Q = \text{diag}(0,10^{-3},10^{-3},10^{-2})$, $R = 1$, $Q_F = \text{diag}(100,100,100,200)$, $\mathcal{V} = \{v\in \mathbb{R}: ||v||_\infty \leq 20\}$ and $K = [0,	-7.3989,	-10.435,	-3.7039]$.
Finally, we implemented the CLF-CBF~\eqref{eq:CLF-CBF-QP} for $\mathcal{S}_e = \{e \in \mathbb{R}^n: e^\top Q_ee \leq 1 \}$ with $Q_e = \text{diag}(1/0.2^2, 1/0.1^2,1/0.05^2,1/0.01^2)$, $\mathcal{X}_c = \mathbb{R}^n$ and $\mathcal{U} = \mathbb{R}$. 
\begin{figure}[h!]
    \centering
	\includegraphics[width= 0.9\columnwidth]{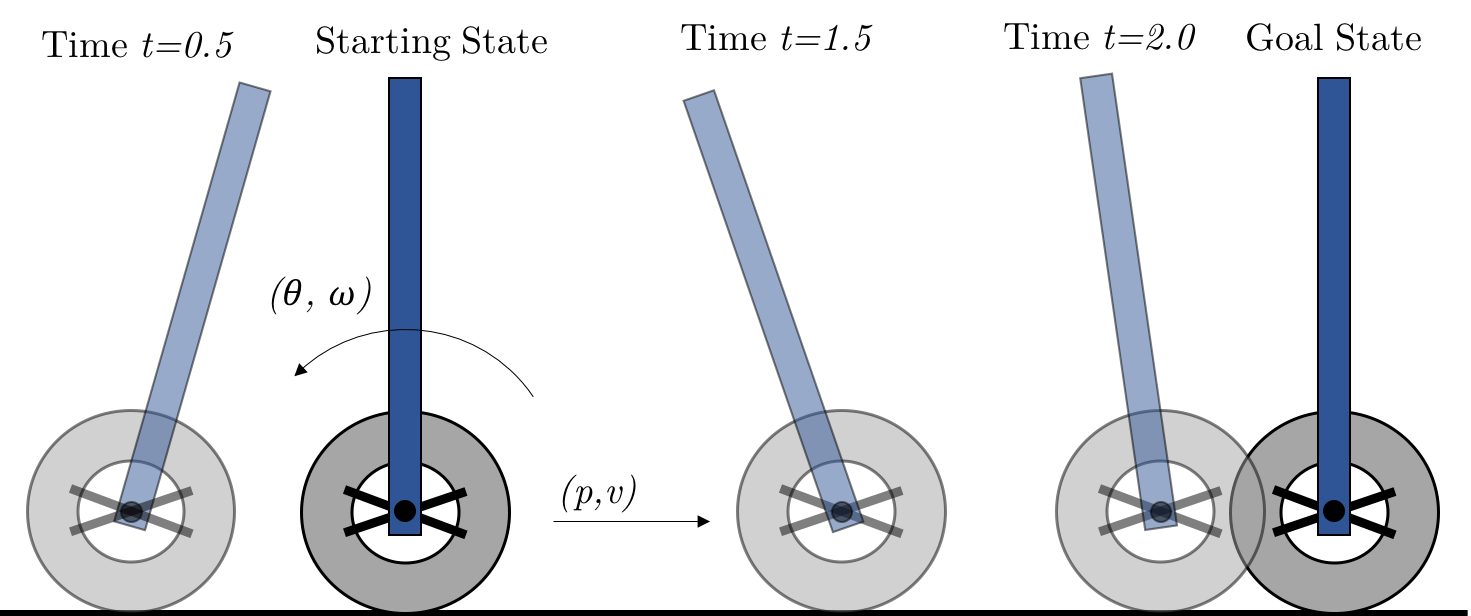}
    \caption{The goal of the controller is to steer the Segway to a goal state while keeping the rod upright. Furthermore, we reported snippets of the state trajectory from Section~\ref{sec:ex1} at $0.5$s,$1.5$s and $2$s.}
    \label{fig:segway}
\end{figure}


\subsection{Unconstrained Example with Low Frequency Update}\label{sec:ex1}

In this example, we run the high level MPC planner at $2$Hz, we set $\mathcal{X}_d = \mathbb{R}^n$ and the MPC horizon $N=10$. Figure~\ref{fig:stateComparison_ex1} shows the closed-loop trajectories for the proposed strategy, a linear MPC and nonlinear MPC policies, which are implemented at $100$Hz and $20$Hz for prediction horizons $N^{\mathrm{100Hz}} = 500$ and $N^{\mathrm{20Hz}} = 100$, respectively.
All strategies plan the desired trajectory over a receding time window of $5$ seconds. 
We notice that the linear MPC overshoots the goal state and the nonlinear MPC discretized at $20$Hz oscillates before reaching the target state. 
On the other hand, the proposed strategy performs similarly to the high frequency nonlinear MPC (discretized at $100$Hz with prediction horizon $N^{\mathrm{100Hz}} = 500$), while being implemented with a $2$Hz model update rate, a prediction horizon $N=10$ and solving convex optimization problems.
This example shows the advantage of the proposed multi-rate architecture, where the high level control action is updated at a lower frequency than the low level input command, as shown in Figure~\ref{fig:input_lowFreq}.


\begin{figure}[h!]
    \centering
	\includegraphics[trim= 10mm 10mm 10mm 10mm, width= 0.85\columnwidth]{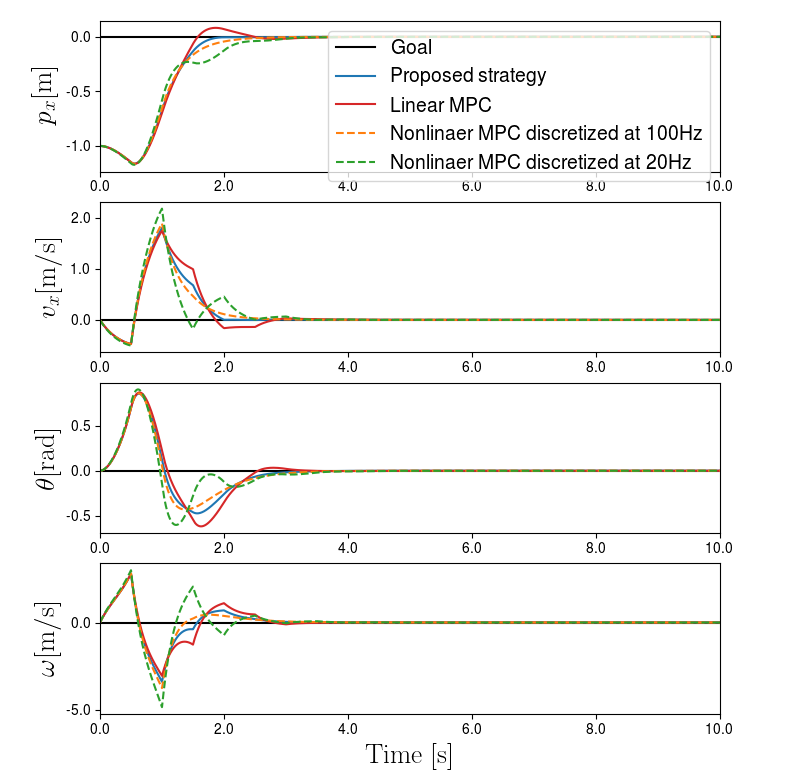}
    \caption{Closed-loop trajectories for the proposed strategy, linear and nonlinear MPC policies. The evolution of the position $p_x$ and velocity $v_x$ underline that the proposed methodology steers the system to the goal state without overshooting.}
    \label{fig:stateComparison_ex1}
\end{figure}



\begin{figure}[h!]
    \centering
	\includegraphics[width= 0.75\columnwidth]{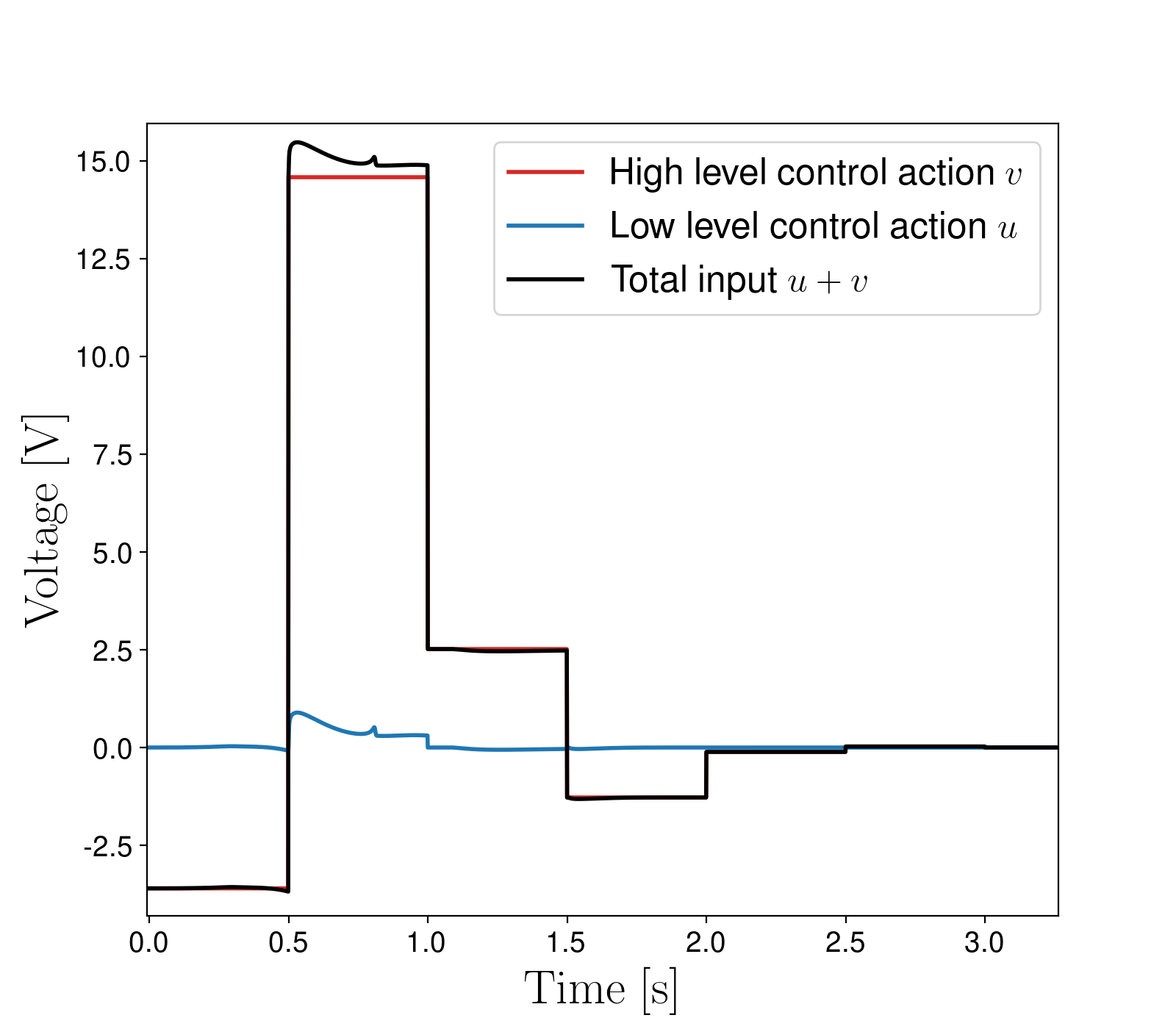}
    \caption{The total input applied to the system is given by the summation of the high level control action from the MPC (in red) and the low level CLF-CBF QP (in blue). We notice that the high level input is updated at $2$Hz, whereas the low level input is updated at $1000$Hz. }
    \label{fig:input_lowFreq}
\end{figure}

\subsection{Constrained Example with High Frequency Update}

In this example, we run the high level MPC planner at $10$Hz, we set $\mathcal{X}_d = \{x = [p_x, v_x, \theta, \omega]^\top \in \mathbb{R}^n : |\theta| \leq 0.78\}$ and the MPC horizon $N=10$.
We compare the proposed strategy with linear MPCs and nonlinear MPCs discretized at $10$Hz, $20$Hz and $100$Hz. Also for linear and nonlinear MPCs, we use the constraint tightening from~\eqref{eq:ftocp}. Figure~\ref{fig:thetaComparisonWithCnstr} shows that when the high frequency input from the  low level controller is not used, the closed-loop system violates the state constraints. We underline that constraint satisfaction for nonlinear MPC policies can be guaranteed using the approaches from~\cite{gao2014tube, kogel2015discrete, yu2013tube, singh2017robust,kohler2020computationally}. However, this example shows the advantage of using the high frequency low level controller to reduce the tracking error. Indeed, when the low level controller is not used, the constraint tightening from~\eqref{eq:ftocp} is not sufficient to guarantee constraint satisfaction, both when linear and nonlinear models are leveraged for planning.
Finally, we underline that the computational cost associated with the proposed strategy is $\sim 0.1$s. Whereas, the computational cost associated with linear MPCs discretized at $10$Hz, $20$Hz and $100$Hz is $\sim 0.1$s,$\sim 0.25$s and $\sim 1$s, respectively.

\begin{figure}[h!]
    \centering
	\includegraphics[width= 0.75\columnwidth]{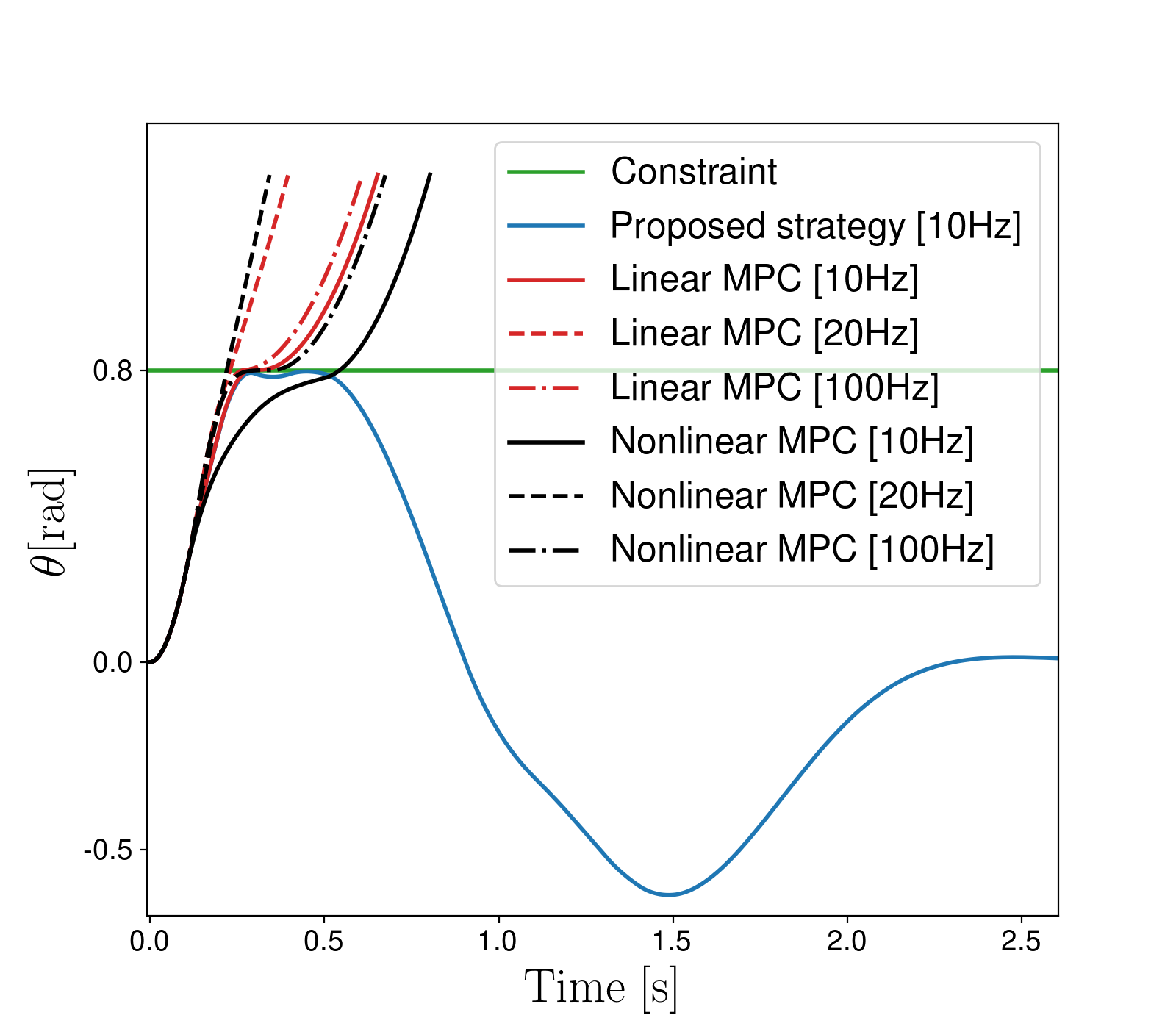}
    \caption{Evolution of the rod angle $\theta$ for the proposed strategy, linear and nonlinear model predictive controllers. When the low level controller is not used the closed-loop system violates the state constraints. }
    \label{fig:thetaComparisonWithCnstr}
\end{figure}

\section{Conclusions}
In this paper, we presented a multi-rate control architecture, where the high level planner and the low level controller run at different frequencies. 
First, we introduced sufficient conditions which guarantee recursive constraint satisfaction for the closed-loop system. Afterwards, we presented a controller design which leverages control barrier functions and MPC policies. 
\renewcommand{\baselinestretch}{0.96}

\bibliographystyle{IEEEtran} 
\bibliography{IEEEabrv,mybibfile}

\end{document}